\documentclass[journal]{IEEEtran}

\usepackage{amsmath,amssymb}
\usepackage{booktabs}
\usepackage{cite}
\usepackage{nicefrac}
\usepackage{graphicx}

\usepackage{xcolor}
\usepackage{url}
\usepackage{mathtools}

\newcommand{\boxast}{\mathbin{\boxplus}}

\newtheorem{theorem}{Theorem}
\newtheorem{definition}{Definition}
\newtheorem{remark}{Remark}

\newcommand{\Fig}[1]{\figurename~\ref{#1}}

\DeclareMathOperator\arctanh{arctanh}
\DeclareMathOperator\sign{sign}

\begin{document}

\title{
Learning LDPC codes with density evolution over relaxed protographs
}

\author{
  \IEEEauthorblockN{Gennady Shutkov, Dmitry Artemasov, Alexey Frolov, Pavel Rybin, Kirill Andreev}

  \IEEEauthorblockA{\textit{Center for Next Generation Wireless and IoT} \\
  \textit{Skolkovo Institute of Science and Technology}\\
  Moscow, Russia \\
  \{g.shutkov, d.artemasov, al.frolov, p.rybin, k.andreev\}@skoltech.ru}
 \thanks{A short version of this paper has been accepted for presentation at IEEE GLOBECOM 2026 Workshops: Workshop on Channel Coding beyond 5G. This extended version of the paper has been submitted to the IEEE Transactions on Communications journal.}
}

\maketitle
\begin{abstract}
We consider the design of low-density parity-check (LDPC) codes for a given iterative decoder. While LDPC performance can be evaluated using simulation, density evolution (DE), or EXIT-chart analysis, selecting a parity-check matrix (PCM) remains a difficult combinatorial optimization problem. Existing approaches often rely on population-based search, random mutations, or genetic algorithms, which require careful tuning and incur high computational cost. Recent gradient descent (GD)-based methods optimize relaxed PCMs by differentiating through decoder simulations, but rely on noisy Monte Carlo estimates, line searches over soft matrix representations, and remain costly for long codes. Moreover, the loss is typically evaluated only at integer-valued PCMs. We focus on long protograph-based LDPC codes and propose a deterministic GD-based framework operating directly on a relaxed protograph representation. The loss is based on DE bit error rate (BER) and can be evaluated directly for relaxed protographs. To justify this relaxation, we associate the relaxed representation with an ensemble of binary PCMs and show that the proposed relaxed DE yields the ensemble-averaged DE performance. The resulting procedure supports standard GD optimization and achieves fast, reliable convergence through deterministic DE evaluation and informative gradients. Numerical results show that the optimized protographs outperform 5G LDPC codes with matching dimensions.
\end{abstract}


\section{Introduction}

Low-density parity-check (LDPC) codes are among the most successful classes of forward-error-correction codes. Since the seminal works of Gallager and Tanner~\cite{Gallager1962,Tanner1981}, their construction has remained a long-standing and practically important research problem. Under belief-propagation (BP), min-sum, and related message-passing decoders, performance depends not only on the code rate and block length but also on the detailed Tanner-graph structure. Consequently, LDPC code design is largely a graph-design problem. A broad range of tools has been developed for analyzing and constructing LDPC codes. Density evolution (DE) predicts iterative-decoding thresholds and enables optimization of irregular ensembles~\cite{RichardsonUrbanke2001,Richardson2001Design}, including recurrent-neural-network-based approaches~\cite{Nisioti2020}, while extrinsic information transfer (EXIT) charts provide a complementary tool for analyzing iterative-decoder convergence~\cite{TenBrink2001}. At finite lengths, graph-construction techniques such as progressive edge growth (PEG)~\cite{Hu2005} improve Tanner-graph connectivity and reduce harmful short cycles. These developments have enabled capacity-approaching and implementation-friendly LDPC constructions~\cite{Chung2001}.

A widely used approach to structured LDPC code design is based on protographs. A small base graph, equivalently represented by a base matrix, specifies the local connectivity pattern and is subsequently lifted to obtain a full-length Tanner graph~\cite{Thorpe2003,Divsalar2009}. The protograph therefore captures the essential structure of the code, while lifting determines its finite-length realization. Protograph EXIT (P-EXIT) analysis enables efficient performance evaluation and threshold optimization directly at the protograph level~\cite{Liva2007}. During lifting, techniques such as approximate cycle extrinsic message degree (ACE)~\cite{Tian2004} can be used to improve the finite-length Tanner-graph structure by suppressing harmful short cycles. This separation between protograph design and lifting makes protograph-based constructions particularly attractive for long structured LDPC codes, since the main design problem can be formulated over a comparatively small base matrix.

In this work, we focus on protograph design. Despite substantial progress in protograph analysis and construction, choosing a protograph suited to a prescribed decoder remains a challenging combinatorial optimization problem. Practical design often relies on discrete optimization over integer-valued base matrices. Typical approaches include local random search, in which one or several matrix entries are randomly modified and the resulting candidate is evaluated, as well as simulated annealing and related mutation-based heuristics that accept or reject such modifications according to a search rule. Population-based approaches have also been proposed, including decoder-in-the-loop genetic optimization of LDPC matrices~\cite{Elkelesh2019}. A key drawback of these methods is that each candidate requires a separate performance evaluation, which can be expensive when the objective is based on DE or decoder simulation. Moreover, random modifications provide no directional information: the search can only compare generated candidates rather than update the protograph in a direction predicted to improve decoding performance.

Recent work has therefore introduced gradient-based optimization into LDPC code design. Differentiable factor-graph optimization for BP decoding was proposed in~\cite{Choukroun2024}, followed by gradient-quantized learning of linear block codes in~\cite{Dufrene2025}. These approaches are promising because gradients provide directional information in an otherwise discrete search space. However, they do not fully exploit the relaxed parity-check matrix (PCM) as an object for direct performance evaluation. In particular, performance is ultimately evaluated for integer-valued matrices, while the relaxed representation primarily serves as an optimization device that must eventually be projected or quantized. Moreover, the training objective is typically estimated using sampled log-likelihood ratios (LLRs) and decoder simulations. This makes the optimization noisy and computationally expensive, especially for long codes, where the number of trainable graph parameters is large and many samples are required to obtain a stable training signal.

In this paper, we propose a deterministic gradient-based framework for protograph optimization that replaces random combinatorial modifications with gradient-guided updates. The method combines deterministic density evolution (DE) with gradient descent (GD) and, unlike decoder-in-the-loop approaches based on Monte Carlo simulations, does not require sampling decoder inputs, thereby avoiding sampling-induced noise. A key feature of the framework is that the loss is evaluated directly on a relaxed protograph. To justify this relaxation, we associate each relaxed protograph with an ensemble of binary parity-check matrices (PCMs). We introduce a relaxed DE procedure and show that its predicted bit error rate (BER) equals the ensemble-averaged DE performance of the associated PCM ensemble. Thus, the relaxed protograph has a direct performance interpretation rather than serving merely as an auxiliary optimization variable. The resulting procedure supports standard
GD optimization and achieves fast and reliable convergence. An interesting empirical observation is that the optimized relaxed protographs are nearly discrete, with only a small number of fractional entries; this behavior is further analyzed in the appendix.

Numerical results demonstrate the applicability of the framework to different decoder objectives and code rates. We consider two protographs corresponding to rates $R=1/2$ and $R=1/5$, using the matching 5G LDPC base graphs~\cite{TS38212} as references. For each rate, we separately optimize a protograph for the sum-product algorithm (SPA) and for layered normalized min-sum (NMS) decoding; in the latter case, the column-wise normalization coefficients are optimized jointly with the protograph. At a block error rate (BLER) of $10^{-3}$, the SPA-optimized protographs outperform the corresponding 5G references by approximately $0.05$ dB for both rates. The NMS-optimized designs achieve gains of approximately $0.12$ dB for $R=1/2$ and $0.18$ dB for $R=1/5$. These results demonstrate that the proposed framework can adapt the protograph to different decoding objectives. More generally, it can be extended to other message-passing decoders and channel models by modifying the corresponding DE updates and input LLR distributions.

The remainder of the paper is organized as follows. Section~\ref{sec:preliminaries} reviews protograph-based LDPC codes, the channel model, and the message-passing decoders considered in this work. Section~\ref{sec:qde} introduces the relaxed protograph ensemble, develops the corresponding density-evolution framework, and formulates the gradient-based protograph optimization problem. Section~\ref{sec:num_results} presents the experimental setup and compares the optimized protographs with the corresponding 5G LDPC reference designs. Section~\ref{sec:conclusion} concludes the paper. The appendices provide implementation details of the differentiable density-evolution procedure and further analysis of the protograph optimization process.

\emph{Notation}. Throughout this paper, capital bold symbols, such as $\mathbf{H}$, denote matrices, while lowercase bold symbols, such as $\boldsymbol{\omega}$, denote vectors; whether a vector is a row or column vector is clear from the context. We define probability measures by san-serif fonts, e.g., $\mathsf P$. We define the integer index set $[n]={1,\ldots,n}$. Calligraphic letters, such as $\mathcal{H}$, denote sets. The symbol $\nabla$ denotes the gradient operator, and we write $\overline{x}=1-x$.
\section{Preliminaries}
\label{sec:preliminaries}

\subsection{Protograph LDPC codes}

LDPC codes are linear block codes defined by sparse parity-check matrices. We consider single-edge protograph-based constructions~\cite{Thorpe2003}, specified by a binary base matrix $\mathbf{H}\in\{0,1\}^{m\times n}$, whose rows and columns correspond, respectively, to check-node and variable-node types. Throughout the paper, $\mathbf{H}$ denotes the protograph (base) matrix.

The lifted matrix $\mathbf{H}^{(Z)}$ is obtained by replacing each nonzero entry of $\mathbf{H}$ by a $Z\times Z$ permutation matrix and each zero entry by a $Z\times Z$ all-zero matrix. If all permutation matrices are circulant permutation matrices, the resulting code is quasi-cyclic (QC). The corresponding binary linear code is
\[
\mathcal{C}
=
\left\{
\mathbf{c}\in\mathbb{F}_2^{nZ}:
\mathbf{H}^{(Z)}\mathbf{c}=\mathbf{0}
\right\},
\]
where all parity-check operations are over $\mathbb{F}_2$.

The 5G LDPC codes are standardized protograph-based QC-LDPC constructions~\cite{TS38212}. They are based on two standardized base graphs together with prescribed lifting, puncturing, and rate-matching procedures.

\subsection{Channel model}

A codeword $\mathbf{c}\in\mathcal{C}$ is BPSK-modulated and transmitted over a binary-input additive white Gaussian noise (BI-AWGN) channel:
\begin{equation}
\label{eq:awgn}
y_i = \phi(c_i)+z_i,
\quad
\phi(c_i)=1-2c_i,
\quad
z_i\overset{\mathrm{i.i.d.}}{\sim}\mathcal{N}(0,\sigma^2).
\end{equation}
The BPSK symbols have unit energy, $E_s=1$. We report the signal-to-noise ratio as $E_s/N_0$, using the standard AWGN normalization
\[
\gamma = \frac{E_s}{N_0}=\frac{1}{2\sigma^2},
\]
where $N_0$ is the noise power spectral density.

The channel log-likelihood ratio (LLR) for bit $i$ is defined as
\[
L_i
=
\log\frac{p(y_i\mid c_i=0)}{p(y_i\mid c_i=1)}
=
\frac{2y_i}{\sigma^2}.
\]

The performance is measured in terms of bit error rate (BER) and block error rate (BLER). Let \(\hat{\mathbf c} =(\hat c_1,\ldots,\hat c_{nZ})\) be the estimated codeword. Then,
\[
P_b\!\left(\mathbf H^{(Z)}\right)
=
\frac{1}{nZ}\sum_{i=1}^{nZ}
\Pr\!\left[c_i\neq \hat c_i\right],
\]
and
\[
P_B\!\left(\mathbf H^{(Z)}\right)
=
\Pr\!\left[\mathbf c\neq \hat{\mathbf c}\right],
\]
respectively.

\subsection{Belief-propagation iterative decoder}

Let $\mathbf{H}^{(Z)}\in\{0,1\}^{mZ\times nZ}$ be a lifted parity-check matrix obtained from the protograph matrix $\mathbf{H}$. Its Tanner graph~\cite{Tanner1981} is bipartite: variable node $i\in[nZ]$ represents code bit $c_i$, check node $j\in[mZ]$ represents the $j$-th parity constraint, and an edge $(i,j)$ exists if and only if $H^{(Z)}_{ji}=1$. Let $\mathfrak{M}(i)$ denote the set of check nodes adjacent to variable node $i$, and let $\mathfrak{N}(j)$ denote the set of variable nodes adjacent to check node $j$.

Belief propagation (BP) operates by exchanging LLR messages along the edges of the Tanner graph. We denote variable-to-check messages by $Q_{i\to j}$ and check-to-variable messages by $R_{j\to i}$. The decoder is initialized by assigning the channel LLR to every outgoing variable-node edge:
\begin{equation}
\label{eq:bp_initialization}
Q_{i\to j}^{(0)}=L_i,
\qquad
j\in\mathfrak{M}(i).
\end{equation}
The equations below define the local scalar node-update rules; their execution order depends on the decoding schedule.

\paragraph{Check-node operation}
For sum-product BP, check node $j$ sends to variable node $i$ the extrinsic message
\begin{equation}
R_{j \to i}
=
2 \arctanh\left(
\prod_{\substack{i'\in\mathfrak{N}(j)\\ i'\neq i}}
\tanh\left(\frac{Q_{i'\to j}}{2}\right)
\right).
\label{eq:bp_check_rule}
\end{equation}
A commonly used low-complexity approximation is normalized min-sum (NMS):
\begin{equation}
R_{j \to i}
=
\alpha_{i,j}
\prod_{\substack{i'\in\mathfrak{N}(j)\\ i'\neq i}}
\sign\left(Q_{i'\to j}\right)
\min_{\substack{i'\in\mathfrak{N}(j)\\ i'\neq i}}
\left|Q_{i'\to j}\right|,
\label{eq:ms_check_rule}
\end{equation}
where $\alpha_{i, j}>0$ is the corresponding normalization coefficient.

\paragraph{Variable-node operation}
The variable-node rule is the same for sum-product and min-sum decoding. Variable node $i$ sends to check node $j$ the channel LLR plus all incoming check-to-variable messages except the one from the target check node:
\begin{equation}
Q_{i \to j}
=
L_i
+
\sum_{\substack{j'\in\mathfrak{M}(i)\\ j'\neq j}}
R_{j'\to i}.
\label{eq:bp_variable_rule}
\end{equation}
After the last decoding iteration, the a posteriori output LLR is
\begin{equation}
\label{eq:bp_output_llr}
L_i^{\mathrm{out}}
=
L_i
+
\sum_{j\in\mathfrak{M}(i)}
R_{j\to i}.
\end{equation}
The decoded bit is determined from the sign of $L_i^{\mathrm{out}}$. The density-evolution procedure in Section~\ref{sec:qde} uses the same scalar node-update rules while tracking their induced probability laws rather than individual LLR messages.

\section{Gradient Descent with a Relaxed Protograph}
\label{sec:qde}

\subsection{Relaxed protograph ensemble}
We first define the relaxed protograph-based LDPC ensemble represented by a single relaxed protograph matrix.

\begin{definition}[Relaxed protograph-based LDPC ensemble]
\label{def:proto_ensemble}
Let
\[
\boldsymbol{\Omega}=(\omega_{ij})\in[0,1]^{m\times n}
\]
be a relaxed protograph matrix, and let \(Z\) be a lifting factor. The relaxed protograph-based LDPC ensemble associated with \(\boldsymbol{\Omega}\) and \(Z\), denoted by
\[
\mathcal E^{(Z)}(\boldsymbol{\Omega}),
\]
is the ensemble of binary parity-check matrices
\[
\mathbf H^{(Z)}\in\{0,1\}^{mZ\times nZ}
\]
partitioned into \(m\times n\) blocks
\[
\mathbf H^{(Z)} =
\begin{pmatrix}
\mathbf H_{11} & \cdots & \mathbf H_{1n}\\
\vdots & \ddots & \vdots\\
\mathbf H_{m1} & \cdots & \mathbf H_{mn}
\end{pmatrix},
\qquad
\mathbf H_{ij}\in\{0,1\}^{Z\times Z}.
\]
Each block \(\mathbf H_{ij}\) is chosen in the following way:
\begin{equation}
\label{eq:fractional_lifting_rule}
\mathbf H_{ij} = \mathbf P_{ij}\odot \mathbf T_{ij},
\end{equation}
where \(\mathbf P_{ij}\) is chosen uniformly at random from the set of all $Z \times Z$ permutation matrices, \(\odot\) denotes the element-wise product, and $\mathbf T_{ij}$ is a $Z\times Z$ binary random matrix with independent Bernoulli entries of parameter $\omega_{ij}$, independent of $\mathbf P_{ij}$, i.e. each nonzero entry of \(\mathbf P_{ij}\) is independently retained with probability \(\omega_{ij}\).
\end{definition}

Therefore, each matrix $\boldsymbol{\Omega}=(\omega_{ij})\in[0,1]^{m\times n}$ defines an \emph{ensemble} of codes from Definition~\ref{def:proto_ensemble}. When $\boldsymbol{\Omega}\in\{0,1\}^{m\times n}$, the construction reduces to the standard protograph-based LDPC ensemble.

Density evolution is performed in the asymptotic regime, i.e., as \(Z\to\infty\). We optimize the relaxed protograph matrix using the DE-predicted BER of the ensemble in Definition~\ref{def:proto_ensemble} as the objective. The optimization is performed over the continuous domain $[0,1]^{m\times n}$, whereas the final code construction uses a binary base matrix. In our experiments (see Section~\ref{sec:num_results}), most entries converge to the binary boundaries, leaving at most \(r=3\) fractional entries. For these entries, we considered the general lifting rule~\eqref{eq:fractional_lifting_rule}; however, to retain the standard protograph-based code construction, we instead evaluate all \(2^r\) possible binarizations and select the best final candidate, whose predicted and simulated performance is very close to that of the solution with fractional entries. The observed concentration of the optimized entries at the binary boundaries is further discussed in Appendix~\ref{app:boundary_concentration}.

\subsection{DE over relaxed protograph}

Consider the DE over relaxed protograph used by the optimizer. The scalar decoder from Section~\ref{sec:preliminaries} is replaced by an equivalent distributional recursion in which each LLR message is represented by its probability law. Incoming messages at a protograph node are treated as independent random variables, as in standard density evolution. This corresponds to the usual tree-like-neighborhood assumption for the lifted Tanner graph. The DE recursion further assumes a flooding decoding schedule. Punctured variable-node types are initialized by the deterministic zero-LLR law, whereas transmitted variable-node types are initialized by the channel LLR law under the zero-codeword assumption.

The DE procedure is independent of the specific representation of the probability distributions. We treat a probability distribution abstractly as a probability measure, independently of its representation as a PMF, a density, a CDF, samples, or a numerical approximation. Let $\mathcal X$ denote the scalar message alphabet and $\mathcal P(\mathcal X)$ the set of probability measures on $\mathcal X$. For ordinary LLR messages, $\mathcal X=\mathbb R$; a formal neutral symbol may be adjoined to this alphabet when required, as discussed below. Let $\mathsf P_i\in\mathcal P(\mathcal X)$ denote the probability law of the $i$-th incoming message.

Because node updates can have different arities, let $f_c^{(D)}:\mathcal X^D\rightarrow\mathcal X$ denote the scalar check-node update with $D$ inputs. When the arity is clear from the arguments, we omit the superscript and write $f_c$. The induced distributional check-node operator $\mathcal F_c$ maps independent input laws to the law of the scalar output. Thus, $\mathsf P_o = \mathcal F_c(\mathsf P_1,\ldots,\mathsf P_D)$ is defined, for every measurable set $A\subseteq\mathcal X$, by
\[
\mathsf P_o(A) =
\int_{\mathcal X^D}
\mathbf 1\!\left\{
 f_c(x_1,\ldots,x_D)\in A
\right\}
\prod_{i=1}^{D}\mathsf P_i(dx_i).
\]
Equivalently, the output law can be written as
\begin{equation}
\mathsf P_o(dx)
=
\int_{\mathcal X^D}
\delta_{f_c(x_1,\ldots,x_D)}(dx)
\prod_{i=1}^{D}\mathsf P_i(dx_i),
\label{eq:random_operator}
\end{equation}
where $\delta_y$ denotes the Dirac probability measure concentrated at $y$. The variable-node distributional operator $\mathcal F_v$ is defined analogously from the corresponding scalar update function $f_v$.

For probability laws $\mathsf P_1$ and $\mathsf P_2$, and weights $a,b\ge 0$ satisfying $a+b=1$, their mixture is the probability law
\[
\mathsf P=a\mathsf P_1+b\mathsf P_2,
\qquad
\mathsf P(A)=a\mathsf P_1(A)+b\mathsf P_2(A)
\]
for every measurable set $A$. Equation~\eqref{eq:random_operator} is mixture-linear in each input law; in particular,
\[
\mathcal F_c(a\mathsf P+b\mathsf Q,\mathsf P_2,\ldots)
=
a\mathcal F_c(\mathsf P,\mathsf P_2,\ldots)
+b\mathcal F_c(\mathsf Q,\mathsf P_2,\ldots).
\]

We next define the neutral elements used to represent absent edges. For a family of scalar update functions $\{f^{(D)}\}_{D\ge 1}$, a value $x_\circ\in\mathcal X$ is called \emph{neutral} if, for every admissible arity and all inputs,
\[
f^{(D+1)}(x_\circ,x_1,\ldots,x_D)=f^{(D)}(x_1,\ldots,x_D).
\]
Since the node-update rules considered here are permutation invariant, the position of the neutral argument is immaterial. The Dirac law $\delta_{x_\circ}$ is then the corresponding \emph{neutral distribution}; by~\eqref{eq:random_operator}, inserting it into the induced distributional operator leaves the output law unchanged.

For the variable-node update~\eqref{eq:bp_variable_rule}, the neutral value is zero, and we denote the corresponding neutral distribution by $\mathsf P_\delta=\delta_0$. For the check-node update, let $x_\infty$ denote a neutral check-node input whenever such a value belongs to the chosen scalar message alphabet, and let $\mathsf P_\infty=\delta_{x_\infty}$. With this notation,
\begin{align}
\mathcal{F}_c\left(\mathsf{P}_\infty, \mathsf{P}_1, \ldots\right) &= \mathcal{F}_c\left(\mathsf{P}_1, \ldots\right),
\label{eq:check_neutral}\\
\mathcal{F}_v\left(\mathsf{P}_\delta, \mathsf{P}_1, \ldots\right) &= \mathcal{F}_v\left(\mathsf{P}_1, \ldots\right).
\label{eq:vnop_neutral}
\end{align}

\begin{remark}
For a clipped min-sum update~\eqref{eq:ms_check_rule} with message alphabet $[-L_c,L_c]$, the positive clipping value $x_\infty=L_c$ is an exact neutral check-node input. For the exact sum-product update~\eqref{eq:bp_check_rule} on real-valued LLRs, no finite real neutral value exists: a neutral argument would have to make the corresponding hyperbolic-tangent factor equal to one. To avoid introducing the extended real line, we may instead adjoin a formal symbol $\varepsilon_c\notin\mathbb R$ to the scalar message alphabet and define it to be neutral by the arity-reduction convention
\[
f_c^{(D+1)}(\varepsilon_c,x_1,\ldots,x_D)
:=f_c^{(D)}(x_1,\ldots,x_D).
\]
In this case, $\mathsf P_\infty:=\delta_{\varepsilon_c}$ in~\eqref{eq:check_neutral}. We additionally set $f_c^{(1)}(\varepsilon_c):=\varepsilon_c$, so a check node with no other active neighbors produces the same formal state. When this state is passed to a variable-node update, it acts as a positive perfect-information state and is absorbing under LLR addition, i.e., $f_v(\varepsilon_c,x_1,\ldots):=\varepsilon_c$. Thus the recursion remains closed on the augmented alphabet $\mathbb R\sqcup\{\varepsilon_c\}$ without identifying $\varepsilon_c$ with an extended real number. The symbol $\varepsilon_c$ is therefore a bookkeeping state, not a finite LLR approximation to $+\infty$. Specific distribution representations and their corresponding low-complexity implementations are discussed in Appendix~\ref{app:de_bin_matrix}.
\end{remark}

Let $\boldsymbol{\omega}=(\omega_1,\ldots,\omega_D)$ denote a local relaxed neighborhood rather than the whole protograph matrix. For a check-node update, $\boldsymbol{\omega}$ is a row of $\boldsymbol{\Omega}$ and $D=n$. For a variable-node update, $\boldsymbol{\omega}$ is a column of $\boldsymbol{\Omega}$ and $D=m$. Thus, $\omega_k$ always denotes the $k$-th entry of the currently considered local row or column; the global row or column index is suppressed only to keep the notation readable.

We first consider the check-node update. For a Bernoulli edge indicator $I\in\{0,1\}$, define the conditional input law and its relaxed mixture by
\begin{equation}
\label{eq:check_relaxed_input}
\hat{\mathsf{P}}_{(k,I)}=
\begin{cases}
\mathsf{P}_k, & I=1,\\
\mathsf{P}_\infty, & I=0,
\end{cases}
\quad
\hat{\mathsf{P}}_k=\omega_k\mathsf{P}_k+\overline{\omega}_k\mathsf{P}_\infty,
\end{equation}
where $\overline{\omega}_k = 1-\omega_k$. By~\eqref{eq:check_neutral}, replacing an absent non-target edge by $\mathsf P_\infty$ leaves the check-node output unchanged.

Applying mixture-linearity of~\eqref{eq:random_operator} together with~\eqref{eq:check_neutral} gives
\begin{equation}
\label{eq:fc_collapse}
\omega_1\mathcal{F}_c(\mathsf{P}_1,\mathsf{P}_2,\ldots) + \overline{\omega}_1\mathcal{F}_c(\mathsf{P}_\infty,\mathsf{P}_2,\ldots) = \mathcal{F}_c(\hat{\mathsf{P}}_1,\mathsf{P}_2,\ldots).
\end{equation}
Using the law of total probability, the marginal law of the $k$-th outgoing message is
\begin{equation}
\label{eq:mixture_full}
\tilde{\mathsf{P}}_k
=\omega_k\sum_{\boldsymbol{\iota}\in\{0,1\}^{D-1}}
\mathbb{P}\{\boldsymbol{\iota}\}
 \mathcal{F}_c\left(\mathcal{L}_{k}(\boldsymbol{\iota}),
 \mathcal{R}_{k}(\boldsymbol{\iota})\right)
+\overline{\omega}_k\mathsf{P}_\delta,
\end{equation}
where $\boldsymbol{\iota}=(i_1,\ldots,i_{k-1},i_{k+1},\ldots,i_D)$ lists the Bernoulli states of all non-target edges, and
\begin{equation}
\label{eq:mixture_sets}
\begin{aligned}
\mathcal{L}_{k}(\boldsymbol{\iota})
&=\bigl(\hat{\mathsf{P}}_{(1,i_1)},\ldots,
\hat{\mathsf{P}}_{(k-1,i_{k-1})}\bigr),\\
\mathcal{R}_{k}(\boldsymbol{\iota})
&=\bigl(\hat{\mathsf{P}}_{(k+1,i_{k+1})},\ldots,
\hat{\mathsf{P}}_{(D,i_D)}\bigr),
\end{aligned}
\end{equation}
with
\begin{equation}
\label{eq:mixture_probability}
\mathbb{P}\{\boldsymbol{\iota}\}
=\prod_{j\in[D]\setminus\{k\}}\omega_j^{i_j}\left(\overline{\omega}_j\right)^{1-i_j}.
\end{equation}
Here $\mathcal F_c(\mathcal L_k,\mathcal R_k)$ denotes the check-node operator applied to the concatenation of the two ordered tuples in~\eqref{eq:mixture_sets}. The leading factor $\omega_k$ in~\eqref{eq:mixture_full} accounts for the target edge itself. If that edge is absent, no check-to-variable message is present; for the purpose of the relaxed recursion it is represented by $\mathsf P_\delta=\delta_0$, the neutral law for the subsequent variable-node sum.

The sum in~\eqref{eq:mixture_full} contains $2^{D-1}$ terms. Repeated application of~\eqref{eq:fc_collapse} pushes the Bernoulli mixtures into the individual inputs and eliminates the explicit enumeration. Therefore,
\begin{equation}
\label{eq:mixture_reduced}
\tilde{\mathsf{P}}_k
=\omega_k \mathcal F_c(\hat{\mathsf{P}}_1,\ldots,\hat{\mathsf{P}}_{k-1},
\hat{\mathsf{P}}_{k+1},\ldots,\hat{\mathsf{P}}_D)
+\overline{\omega}_k\mathsf{P}_\delta.
\end{equation}
Thus, the relaxed check-node update gives the exact marginal output law over the local Bernoulli edge ensemble without explicitly enumerating its $2^{D-1}$ configurations.

The variable-node case uses the same local notation, but now $\boldsymbol{\omega}$ is a column of $\boldsymbol{\Omega}$ and $D=m$. By~\eqref{eq:vnop_neutral}, an absent incoming check-to-variable edge is represented by $\mathsf P_\delta$. Hence, each non-target input law is replaced by
\begin{equation}
\label{eq:variable_relaxed_input}
\check{\mathsf{P}}_j=\omega_j\mathsf{P}_j+\overline{\omega}_j\mathsf{P}_\delta.
\end{equation}
The extrinsic variable-to-check message sent along the $k$-th local edge is then
\begin{equation}
\label{eq:variable_relaxed_update}
\tilde{\mathsf{P}}_k
=
\mathcal{F}_v(\mathsf{P}_{\mathrm{ch}}, \check{\mathsf{P}}_1,\ldots,\check{\mathsf{P}}_{k-1},
\check{\mathsf{P}}_{k+1},\ldots,\check{\mathsf{P}}_D),
\end{equation}
where $\mathsf{P}_{\mathrm{ch}}$ is the channel LLR law. Unlike~\eqref{eq:mixture_reduced}, no leading factor $\omega_k$ is required because a variable-to-check message associated with an absent target edge is not consumed by the subsequent check-node update.

To evaluate the DE-predicted bit error rate (BER), we construct the a posteriori LLR law of a variable node after the prescribed number of DE iterations, omitting the variable-node index for simplicity. Let $f_o$ denote the scalar output-LLR sum and $\mathcal F_o$ its induced distributional operator. Then
\begin{equation}
\label{eq:ber_out_relaxed}
\mathsf{P}^{\mathrm{out}} = \mathcal{F}_o\left(\mathsf{P}_{\mathrm{ch}}, \tilde{\mathsf{P}}_1, \ldots, \tilde{\mathsf{P}}_m\right), 
\end{equation}
$$
f_o(x_1, \ldots, x_{m+1}) = \sum_{i=1}^{m+1} x_i.
$$
Here $\tilde{\mathsf{P}}_i$ are the check-node output laws from~\eqref{eq:mixture_reduced}. Equation~\eqref{eq:ber_out_relaxed} is the DE counterpart of~\eqref{eq:bp_output_llr}.

Under the zero-codeword assumption, the performance metric is the variable-node-averaged BER,
\begin{equation}\label{eq:de_ber}
\hat{P}_b\left(\boldsymbol{\Omega}\right) = \frac{1}{n} \sum\limits_{i = 1}^{n}\int_{\mathbb{R}^-}\mathsf P_i^{\mathrm{out}}(dx_i).
\end{equation}
The preceding construction gives the following result.
\begin{theorem}
Let \(\boldsymbol{\Omega}\in[0,1]^{m\times n}\). For any fixed number \(\ell\) of decoding iterations, under the density-evolution assumptions stated above,
\[
\lim_{Z\to\infty}
\mathbb E_{\mathbf H^{(Z)}\sim\mathcal E^{(Z)}(\boldsymbol{\Omega})}
\left[
P_b^{(\ell)}\!\left(\mathbf H^{(Z)}\right)
\right]
=
\hat{P}_b^{(\ell)}(\boldsymbol{\Omega}).
\]
\end{theorem}

\begin{remark}
The superscript \(\ell\) indicates that both the BER and the DE metric are evaluated after \(\ell\) decoding iterations.
\end{remark}

The relaxed DE recursion defines the optimization objective directly on $[0,1]^{m\times n}$:
\begin{equation}
\label{eq:relaxed_de_optimization}
\boldsymbol{\Omega}^{\star}
\in
\operatorname*{arg\,min}_{\boldsymbol{\Omega}\in[0,1]^{m\times n}}
\hat{P}_b(\boldsymbol{\Omega}),
\end{equation}
where $\hat{P}_b(\boldsymbol{\Omega})$ is given by~\eqref{eq:de_ber}. 

\begin{remark}
The low-complexity quantized DE implementation is differentiable with respect to the entries $\omega_{ij}$, so~\eqref{eq:relaxed_de_optimization} can be addressed by standard gradient-based optimization. See Appendix~\ref{app:de_bin_matrix} for the implementation details.
\end{remark}

\section{Numerical Results}
\label{sec:num_results}
\subsection{Experimental setup}
We evaluate the proposed optimization method over the binary-input AWGN channel with BPSK modulation~\eqref{eq:awgn}. The experimental cases listed in Table~\ref{tab:experimental_cases} use long 5G LDPC codes based on base graphs~1 and~2 as references. For each target rate, the optimized protograph has the same effective size and the same number of punctured variable-node types as the corresponding 5G construction. The comparison therefore isolates the effect of the optimized edge pattern rather than changes in the code rate, lifting factor, or puncturing overhead.

\begin{table}[t]
\centering
\caption{Code-design cases used in the experiments.}
\label{tab:experimental_cases}
\footnotesize
\begin{tabular}{@{}lccccc@{}}
\toprule
Case & Rate & $(k,n)$ & BG size & Factor $Z$ & Punct.\\
\midrule
BG2-low & $1/5$ & $(3840, 19200)$ & $42\times 52$ & $384$ & $2Z$\\
BG1-medium & $1/2$ & $(8448, 16896)$ & $24\times 46$ & $384$ & $2Z$\\
\bottomrule
\end{tabular}
\end{table}

\begin{table}[t]
\centering
\caption{Density-evolution and training parameters.}
\label{tab:training_parameters}
\begin{tabular}{@{}p{0.35\columnwidth}p{0.57\columnwidth}@{}}
\toprule
Parameter & Value \\
\midrule
Channel model & AWGN with BPSK \\
LLR grid & $L_c=50$, $2001$ points ($N=1000$) \\
Decoder in DE & Flooding NMS or flooding SPA \\
DE iterations & $40$ (NMS), $50$ (SPA) \\
Optimization method & Projected gradient descent \\
Learning-rate policy & Base learning rate $\nu_t = 1\times10^{-3}$ \\
Constraint handling & Hard clipping: $\omega_{ij}\in[0,1]$, $\alpha_j>0$ \\
Initialization & Single-shot from one initial protograph \\
Implementation & PyTorch automatic differentiation \\
Simulation decoders & Layered NMS / flooding SPA \\
Simulated iterations & $15$ (NLMS) / $50$ (SPFA) \\
\bottomrule
\end{tabular}
\end{table}

Experimental parameters are summarized in Table~\ref{tab:training_parameters}. Training uses projected gradient descent with a base learning rate $\nu_t$. To accelerate convergence, the actual learning rate is adapted according to changes in the gradient signs of the trainable protograph parameters $\omega_{ij}$ between consecutive optimization steps, while never falling below the base learning rate. Depending on the number of sign changes, the scheduler increases or decreases the learning rate or restores the previous optimization state.\footnote{See \url{https://github.com/and-kirill/de-proto-trainable}.}

For each target rate, we optimize two separate protographs for the two decoding objectives. The NLMS-targeted design is trained using flooding NMS density evolution while jointly optimizing the relaxed protograph matrix and the column-wise normalization coefficients, as soon as the performance of normalized min-sum decoding depends on the normalization factors in~\eqref{eq:ms_check_rule}. We allow the factor $\alpha_{i,j}$ to depend only on the target protograph column. The coefficients $\boldsymbol{\alpha}=(\alpha_1,\ldots,\alpha_n)$ are optimized jointly with $\boldsymbol{\Omega}$ by backpropagation.

The SPFA-targeted design is trained independently using flooding SPA density evolution. Finite-length evaluation uses layered NLMS decoding with $15$ iterations for the NLMS-targeted design and flooding SPFA decoding with $50$ iterations for the SPFA-targeted design. Consequently, the NLMS and SPFA results reported for the proposed method correspond to different optimized protographs; they should not be interpreted as two decoders applied to the same matrix.

For the 5G NLMS baseline, the base graph is kept fixed and the column-wise coefficients $\boldsymbol{\alpha}$ are optimized. We additionally report a 5G NLMS reference with a single optimized normalization coefficient $\alpha^\star$. Thus, the main NLMS comparison is between a fixed 5G base graph with optimized per-column coefficients and a jointly optimized protograph and coefficient vector. The optimized base graphs are subsequently converted to their final binary form as described in Appendix~\ref{app:boundary_concentration} and lifted using the standard ACE algorithm~\cite{Tian2004}.

\subsection{Trained LDPC decoding performance}

\begin{figure}
\centering

\includegraphics[page=1]{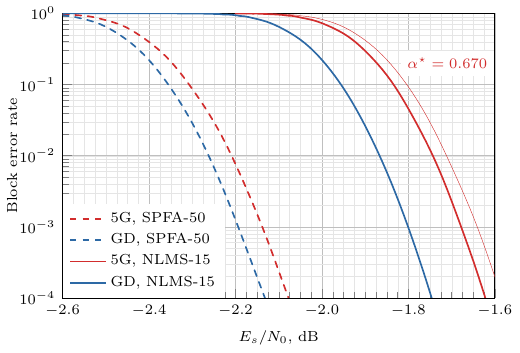}
\includegraphics[page=2]{figs.pdf}

\caption{BLER performance for the BG1-medium case at rate $R=1/2$ (top) and the BG2-low case at rate $R=1/5$ (bottom), see Table~\ref{tab:experimental_cases}. Red curves correspond to the 5G reference codes and blue curves to the optimized designs. Dashed curves show flooding SPFA with $50$ iterations; solid curves show layered NLMS with $15$ iterations. For each rate, the optimized SPFA and NLMS curves correspond to separately trained protographs. The main 5G NLMS baseline uses optimized column-wise coefficients $\alpha_j$, whereas the thin red solid curve uses a single optimized coefficient $\alpha^\star$, whose value is indicated in each panel. The optimized NLMS design is obtained by joint optimization of the protograph and the coefficients $\alpha_j$.}
\label{fig:r50}
\end{figure}

Figure~\ref{fig:r50} shows that the decoder-matched optimized designs outperform the corresponding 5G references for both considered rates and both decoding objectives. At a BLER of $10^{-3}$, the SPFA-targeted protographs provide a gain of approximately $0.05$~dB over the 5G SPFA reference for both $R=1/2$ and $R=1/5$. Under layered NLMS decoding, the jointly optimized protograph and normalization coefficients provide gains of approximately $0.12$~dB for $R=1/2$ and $0.18$~dB for $R=1/5$, relative to the 5G base graphs with their own optimized column-wise coefficients.

The thin red curves in Figure~\ref{fig:r50} show the 5G NLMS performance with a single optimized coefficient $\alpha^\star$; the corresponding values are $0.670$ for $R=1/2$ and $0.657$ for $R=1/5$. Optimizing separate coefficients $\alpha_j$ for the columns of the 5G base graph further improves this baseline, as seen from the thicker red NLMS curves. Therefore, the NLMS gains reported above are not obtained merely by providing the optimized protograph with a more flexible normalization rule: both the 5G reference and the proposed NLMS design use optimized column-wise coefficients, while the latter additionally optimizes the protograph structure.

The SPFA gains are smaller than the NLMS gains but remain consistent for both rates. The larger NLMS improvement at the lower rate is observed empirically; no specific structural cause is inferred from these experiments alone. Since the NLMS- and SPFA-targeted results are obtained from different optimized protographs, Figure~\ref{fig:r50} does not by itself characterize the cross-decoder robustness of either design. Such a comparison would require evaluating each optimized protograph under both decoding rules.

Overall, the results demonstrate that the proposed gradient-based procedure can improve upon decoder-matched 5G references under both NMS- and SPA-based objectives. In the NLMS case, joint optimization of the protograph and the normalization coefficients provides an additional design degree of freedom beyond coefficient optimization on a fixed base graph; in the SPFA case, the same framework applies without introducing normalization parameters.

\section{Conclusions and Future Work}
\label{sec:conclusion}

We presented a decoder-aware framework for gradient-based optimization of protograph-based LDPC codes using density evolution over a relaxed protograph. The proposed formulation provides a deterministic and differentiable objective for different iterative decoders and enables joint optimization of the protograph and decoder-specific parameters without decoder-in-the-loop simulation. Numerical results for SPA and normalized min-sum decoding demonstrated consistent performance gains over the corresponding 5G LDPC reference designs. Future work will focus on improving the optimization objective and convergence by incorporating finite-length effects and additional regularization terms, as well as on extending the framework to jointly optimize multiple protographs or a common protograph for different decoding algorithms.
\appendices
\section{Low-Complexity DE Implementation}
\label{app:de_bin_matrix}
In this section, we consider a low-complexity and differentiable implementation of the proposed DE, which allows us to use standard optimization tools to obtain the optimized protograph.
\subsection{Distribution representation}
We implement the probability laws from Section~\ref{sec:qde} on a clipped, uniformly quantized LLR grid, following the quantized-distribution approach in~\cite{Simegn2025}. Let $L_c=N\Delta$ for an integer $N$ and define
\begin{equation}
\label{eq:pdf_support}
\begin{aligned}
\mathcal{X}_\Delta
&=\left\{x_a=a\Delta:\;a=-N,-N+1,\ldots,N\right\}\\
&=\left[-L_c,-L_c+\Delta,\ldots,0,\ldots,L_c\right].
\end{aligned}
\end{equation}
A probability law $\mathsf P$ represented on this grid is stored as the PMF vector
\begin{equation}
\label{eq:pmf_vector_equivalence}
\mathbf p=(p_{-N},\ldots,p_N)^{\mathsf T},
\qquad
p_a=\mathsf P(\{x_a\}).
\end{equation}
Values outside $[-L_c,L_c]$ are clipped to the nearest boundary. Whenever a scalar node update produces an off-grid value $y\in[x_a,x_{a+1}]$, its probability mass is linearly split between the two adjacent grid points with weights $(x_{a+1}-y)/\Delta$ and $(y-x_a)/\Delta$, respectively. Although the neighboring grid indices are selected by floor/ceiling operations, the resulting interpolation map is continuous and piecewise linear in $y$, and hence differentiable almost everywhere. This property permits backpropagation through the remapping step. The remapping is also linear in the probability masses and therefore preserves the mixture-linearity used in the relaxed DE derivation. We use linear interpolation rather than nearest-neighbor quantization, which produced noticeable discretization artifacts even on fine grids in our numerical tests.

Let $M=|\mathcal X_\Delta|=2N+1$. A direct evaluation of a generic $D$-input node rule scales as $\mathcal O(M^D)$. The implementations below instead use binary reductions and reuse partial results across extrinsic outputs. A generic two-input distributional update costs $\mathcal O(M^2)$; the variable-node sum is implemented by FFT-based convolution, and the min-sum check-node update admits a linear-time implementation in $M$.

\begin{remark}
Previously, we considered a clipping procedure that may violate the associativity of the node-processing functions~\eqref{eq:bp_check_rule}, \eqref{eq:ms_check_rule}, and~\eqref{eq:bp_variable_rule}, and, consequently, invalidate the corresponding binary reduction. Nevertheless, in practice, this issue can be avoided by choosing a sufficiently large value of $L_c$. In our experiments, $L_c \gtrsim 30$ was sufficient.
\end{remark}

For the variable-node rule~\eqref{eq:bp_variable_rule}, addition of independent quantized LLRs corresponds to convolution of their PMFs. We use zero-padded FFTs to evaluate the linear convolution and accumulate probability mass outside $[-L_c,L_c]$ at the corresponding clipping boundary. A binary variable-node update therefore has complexity $\mathcal O(M\log M)$.

For the normalized min-sum (NMS) check update~\eqref{eq:ms_check_rule}, we first evaluate the unscaled sign/minimum reduction and apply the normalization factor only once to each outgoing check-node message. For two independent nonnegative random variables $\chi_1$ and $\chi_2$,
\begin{equation}
\label{eq:ccdf_ms}
F_{\min(\chi_1,\chi_2)}(x)
=1-\left(1-F_{\chi_1}(x)\right)\left(1-F_{\chi_2}(x)\right).
\end{equation}
To account for the sign product in~\eqref{eq:ms_check_rule}, we use sign-split tails. For a PMF vector $\mathbf p$, write
\begin{equation}
\label{eq:sign_split_tails}
\mathbf{p}=(\mathbf{p}^-,p^0,\mathbf{p}^+),
\qquad
\mathbf{c}_{\mathbf{p}}=(\mathbf{c}_{\mathbf{p}}^-,p^0,\mathbf{c}_{\mathbf{p}}^+),
\end{equation}
where $p^0$ is the mass at zero and
\begin{equation}
\label{eq:tail_components}
\begin{aligned}
\relax[\mathbf{c}_{\mathbf{p}}^+]_a &= \sum_{b=a}^{N}p_b,
& a&=1,\ldots,N,\\
\relax[\mathbf{c}_{\mathbf{p}}^-]_{-a} &= \sum_{b=-N}^{-a}p_b,
& a&=1,\ldots,N.
\end{aligned}
\end{equation}
For two independent inputs with PMFs $\mathbf p$ and $\mathbf r$, the unscaled min-sum output satisfies
\begin{equation}
\label{eq:check_sign_cases}
\begin{aligned}
\mathbf{c}_{o}^{+}
&=\mathbf{c}_{\mathbf{p}}^{+}\odot\mathbf{c}_{\mathbf{r}}^{+}
 +\mathbf{c}_{\mathbf{p}}^{-}\odot\mathbf{c}_{\mathbf{r}}^{-},\\
\mathbf{c}_{o}^{-}
&=\mathbf{c}_{\mathbf{p}}^{+}\odot\mathbf{c}_{\mathbf{r}}^{-}
 +\mathbf{c}_{\mathbf{p}}^{-}\odot\mathbf{c}_{\mathbf{r}}^{+},\\
p_o^0&=p^0+r^0-p^0r^0,
\end{aligned}
\end{equation}
where $\odot$ denotes componentwise multiplication. The positive and negative PMF components are recovered by adjacent differences of the corresponding tails. The NMS factor $\alpha$ is then applied once, and the scaled masses are remapped to the grid using the interpolation rule above. The resulting binary check-node update has complexity $\mathcal O(M)$.

For the sum-product check-node rule~\eqref{eq:bp_check_rule}, we use the generic two-input update: for each pair of input bins, the scalar check-node rule is evaluated, the product of the two input masses is assigned to the resulting LLR, and this mass is remapped to the grid. The resulting complexity is $\mathcal O(M^2)$ per binary update. The formal neutral check-node state introduced in Section~\ref{sec:qde} is realized numerically by the boundary value $L_c$. With the clipping level used in our implementation, $\tanh(L_c/2)$ rounds to one in the employed floating-point precision, so the corresponding factor is numerically neutral in~\eqref{eq:bp_check_rule}. This is an implementation convention and does not identify $L_c$ with an exact finite neutral value of the real-valued sum-product rule.

As a numerical validation of the prefix--suffix implementation, we repeated DE for several final optimized protographs after simultaneously permuting the columns of $\boldsymbol{\Omega}$ and the associated column-wise decoder parameters. These permutations preserve the underlying decoding problem while changing the order of the intermediate prefix--suffix reductions. In all tested cases, the resulting DE-predicted BER agreed to at least four significant digits.

\subsection{Training the protograph}
\label{subsec:soft_gradient_optimization}

The relaxed DE recursion is implemented entirely by differentiable tensor operations. The relaxed coefficients enter through the affine mixtures defined in Section~\ref{sec:qde}; on the quantized grid, the remaining operations are FFT-based convolutions, sign-split tail products, and linear remapping. Backpropagation through the unrolled DE iterations therefore gives $\nabla_{\boldsymbol{\Omega}}\ell(\boldsymbol{\Omega})$ without Monte Carlo sampling, where
\[
\ell(\boldsymbol{\Omega})=\log \hat{P}_b(\boldsymbol{\Omega})
\]
and $\hat{P}_b$ is defined in~\eqref{eq:de_ber}. The logarithm leaves the minimizer unchanged and improves the scale of the optimization signal at small BER values.

After each gradient step, the relaxed matrix is projected onto the box constraints,
\begin{equation}
\label{eq:soft_gradient_update}
\boldsymbol{\Omega}^{(t+1)}=
\Pi_{[0,1]}
\left(
\boldsymbol{\Omega}^{(t)}-\nu_t
\nabla_{\boldsymbol{\Omega}}\ell(\boldsymbol{\Omega}^{(t)})
\right),
\end{equation}
where $\nu_t$ is the learning rate and $\Pi_{[0,1]}$ clips each entry to $[0,1]$. For the trainable scales, we apply another projection $\alpha_j > 0$. We use an SNR-continuation schedule: when the DE BER falls below a prescribed threshold, the training SNR is decreased and the protograph parameters are left unchanged for that epoch; otherwise, the projected gradient step~\eqref{eq:soft_gradient_update} is applied. This keeps the DE objective in a numerically informative range. Training is performed in a single run from one initial protograph.

\section{Analysis of the Protograph Learning Procedure}
\label{app:boundary_concentration}

This section analyzes the convergence behavior of the proposed protograph optimization procedure and the boundary concentration observed during learning.  Throughout, we use a rate-$1/2$ code with a $24\times46$ relaxed protograph and two punctured variable-node types as an example (see Table~\ref{tab:experimental_cases}, BG1-medium). Its optimization objective is evaluated by quantized density evolution for SPA with $50$ decoding iterations.  The observed boundary concentration is likely driven by the need to achieve a low output BER. Below, we use local stability conditions to identify some structural requirements relevant to this goal, without claiming a complete characterization.

\subsection{Message recursion and its constant terms}

For a symmetric LLR density $a$ conditioned on $X=+1$, the Bhattacharyya parameter is
\begin{equation*}
\mathcal B(a) =\mathbb E_{L\sim a}\!\left[e^{-L/2}\right] =\int_{\mathbb R}e^{-\ell/2}a(\ell)\,d\ell.
\end{equation*}
It is zero for an \emph{error-free message} and one for a completely uninformative message; $P_{\mathrm e}(a)\leq\mathcal B(a)/2$ with random tie-breaking.  Independent LLR addition and the SPA box-sum satisfy~\cite{RichardsonUrbanke2001,richardson2008modern}
\begin{equation*}
\begin{aligned}
\mathcal B(a*b)&=\mathcal B(a)\mathcal B(b),\\
\mathcal B(a\boxast b)&\leq\mathcal B(a)+\mathcal B(b)-\mathcal B(a)\mathcal B(b).
\end{aligned}
\end{equation*}
We use the channel parameters $\beta_j=\beta$ for transmitted types and $\beta_j=1$ for punctured types.

Let $x_{ij}$ and $y_{ij}$ bound the variable-to-check and check-to-variable parameters on a present edge $(i,j)$, and let $z_j$ be a posteriori probability (APP) bound for variable-node type $j$.  The column vectors $\boldsymbol x=(x_{ij})_{(i,j)\in\mathcal E}$ and $\boldsymbol y=(y_{ij})_{(i,j)\in\mathcal E}$ contain one coordinate per edge in $\mathcal E=\{(i,j):\omega_{ij}>0\}$, ordered first by $i$ and then by $j$.  The vector $\boldsymbol z=(z_1,\ldots,z_n)^{\mathsf T}$ contains one APP coordinate per variable-node type. For independent edge retention, the check, variable, and APP maps are
\begin{equation}
\label{eq:concise_maps}
\begin{aligned}
C_{ij}(\boldsymbol x)& \triangleq 1-\prod_{q\ne j}(1-\omega_{iq}x_{iq}),\\
V_{ij}(\boldsymbol y)& \triangleq \beta_j\prod_{p\ne i}
(1-\omega_{pj}+\omega_{pj}y_{pj}),\\
z_j(\boldsymbol y)& \triangleq \beta_j\prod_p
(1-\omega_{pj}+\omega_{pj}y_{pj}).
\end{aligned}
\end{equation}
For the actual check-to-variable parameter $\widetilde y_{ij}$, the box-sum inequality gives $\widetilde y_{ij}\leq C_{ij}(\boldsymbol x)$, not equality in general. At decoding iteration $\ell$, the bounds evolve as
\begin{equation*}
\begin{aligned}
\boldsymbol y^{(\ell)}&=C(\boldsymbol x^{(\ell)}),\\
\boldsymbol x^{(\ell+1)}&=V(\boldsymbol y^{(\ell)})
=F(\boldsymbol x^{(\ell)}),\\
z_j^{(\ell)}&=z_j(\boldsymbol y^{(\ell)}),
\end{aligned}
\end{equation*}
where $F=V\circ C$, and $\circ$ denotes composition: the check update $C$ is followed by the variable update $V$. The maps are componentwise nondecreasing, so these bounds remain valid when initialized by the channel parameters.  All fixed points and Jacobians below refer to this upper comparison recursion for ideal SPA, not to the finite-grid PMF recursion used for training.

\begin{figure}[!t]
\centering
\includegraphics[page=3]{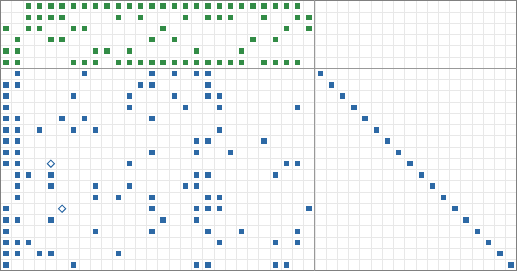}
\caption{Final protograph after the core/leaf permutation.  Green: all nonzero core entries; blue: entries in checks containing leaves.  Squares denote exact units, open diamonds fractional entries, and blank positions zeros.}
\label{fig:learning_core_spy}
\end{figure}

A fixed point satisfies $F(\boldsymbol x^*)=\boldsymbol x^*$. It is locally asymptotically stable if all sufficiently small admissible initial perturbations remain small and decay:
\begin{equation*}
\lim_{\ell\to\infty}\|\boldsymbol x^{(\ell)}-\boldsymbol x^*\|=0.
\end{equation*}
For the error-free candidate $\boldsymbol x^*=\boldsymbol 0$, introduce the derivative matrices
\begin{equation*}
\mathbf{C}_0=\left.\frac{\partial C(\boldsymbol x)}{\partial\boldsymbol x}
\right|_{\boldsymbol x=0},\qquad
\mathbf{V}_0=\left.\frac{\partial V(\boldsymbol y)}{\partial\boldsymbol y}
\right|_{\boldsymbol y=0}.
\end{equation*}
Since $C(0)=0$ and $V(0)=\boldsymbol b$, the two expansions are
\begin{equation*}
\begin{aligned}
C(\boldsymbol x)&=\mathbf{C}_0\boldsymbol x+O(\|\boldsymbol x\|^2),\\
V(\boldsymbol y)&=\boldsymbol b+\mathbf{V}_0\boldsymbol y+O(\|\boldsymbol y\|^2).
\end{aligned}
\end{equation*}
Their composition yields
\begin{equation}
\label{eq:concise_linear}
\begin{aligned}
\boldsymbol x^{(\ell+1)}
&=\boldsymbol b+\mathbf{J}\boldsymbol x^{(\ell)}
+O(\|\boldsymbol x^{(\ell)}\|^2),\\
\mathbf{J}&=\mathbf{V}_0\mathbf{C}_0,\qquad
b_{ij}=\beta_j\prod_{s\ne i}(1-\omega_{sj}).
\end{aligned}
\end{equation}
The order $\mathbf{V}_0\mathbf{C}_0$ follows the check update by the variable update. For $(i,j),(p,q)\in\mathcal E$, the matrix entry is
\begin{equation}
\label{eq:concise_j_entry}
[\mathbf{J}]_{(i,j),(p,q)}
=\beta_j\omega_{pj}\omega_{pq}
\prod_{s\ne i,p}(1-\omega_{sj}),
\end{equation}
when $p\ne i$ and $q\ne j$; all other entries are zero. Here $i,p,s$ are check indices and $j,q$ are variable indices. The APP constant is $a_j=z_j(0)=\beta_j\prod_s(1-\omega_{sj})$. If $\boldsymbol b\ne0$, \eqref{eq:concise_linear} is an affine approximation and the origin is not a fixed point, regardless of $\mathbf{J}$.

\begin{figure}
\centering
\includegraphics[page=4]{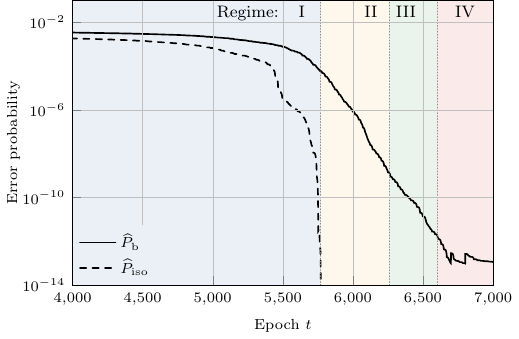}
\caption{Optimization trajectory and its four regimes. Solid line represents saved BER, while the dashed line represents isolation-induced contribution.  The latter is exactly zero from epoch $5771$ and is then omitted from the logarithmic axis.  The first $t_0 = 4000$ epochs are omitted as soon as both lines do not demonstrate any significant changes and proceed parallel. BER spikes in zone IV correspond to the scheduled SNR decrease. Similarly, after $t>7000$, the convergence happened and there were no visible changes in the plot.}
\label{fig:learning_overview}
\end{figure}

\subsubsection{No exact unit in a column}

Then $a_j>0$ at finite SNR.  The variable has a positive probability of being isolated and can then use only its channel observation.  Averaging these contributions under independent edge retention gives the BER lower bound
\begin{equation}
\label{eq:concise_floor}
\widehat P_{\mathrm b}\geq\widehat P_{\mathrm{iso}}
=\frac1n\sum_{j=1}^n\widehat p_{\mathrm{ch},j}
\prod_{r=1}^m(1-\omega_{rj}).
\end{equation}
Here $\widehat p_{\mathrm{ch},j}$ is the channel-only error probability evaluated with the objective used in this paper.  To match this objective, we set $\widehat p_{\mathrm{ch},j}=1$ at punctured positions. This is the isolation-induced part of the floor, not necessarily the complete asymptotic BER. It vanishes exactly when every column has an exact unit.  Where
$\widehat P_{\mathrm b}\approx\widehat P_{\mathrm{iso}}$, reducing this floor determines the achievable performance.

\subsubsection{Exactly one unit in a column}

If $\omega_{ij}=1$ is the only unit in column $j$, then $a_j=0$, but
\begin{equation*}
b_{ij}=\beta_j\prod_{r\ne i}(1-\omega_{rj})>0.
\end{equation*}
The constant in \eqref{eq:concise_linear} therefore still prevents the all-zero extrinsic fixed point.  Ordinary error-free stability requires at least two exact units in every column.  One unit can nevertheless suffice for an error-free APP value, as the following construction shows.

\subsection{Single-unit leaves and the core matrix}

A \emph{single-unit leaf} is a column containing one unit and otherwise only zeros.  Its unique check receives a channel-valued message from the leaf, so it cannot supply error-free messages to its other neighbors at finite SNR.  We separate these checks from the core.

Move the leaf columns to the right and their checks to the bottom, pairing each leaf with its check.  These checks must be distinct. Let $\mathcal I_{\mathrm L}$ be the set of row indices of checks containing leaves. For each $i\in\mathcal I_{\mathrm L}$, let $j_i$ be the column index of its unique single-unit leaf. The remaining row indices form $\mathcal I_c$; all remaining column indices form $\mathcal{J}_c$.  In particular, no non-leaf neighbor of a bottom check is omitted.  For $r$ leaves, this permutation gives
\begin{equation}
\label{eq:concise_core_block}
\widetilde{\boldsymbol{\Omega}}
=\begin{pmatrix*}[l]
\boldsymbol{\Omega}_c&\mathbf 0\\
\boldsymbol{\Omega}_{\mathrm L c}&\mathbf I_r
\end{pmatrix*}.
\end{equation}
Both left blocks are restrictions of the relaxed matrix $\boldsymbol{\Omega}$; their entries retain the notation $\omega_{ij}$ and the original indices. The core consists of rows $\mathcal I_c$ and columns $\mathcal{J}_c$. Our structural requirement is at least two exact units in every column of $\boldsymbol{\Omega}_c$.  The block organization resembles that of 5G NR LDPC codes~\cite{Cui2020Core}, but this two-unit requirement is specific to our
stability analysis.

Figure~\ref{fig:learning_core_spy} shows the final optimized matrix after row and column permutations giving the block form
\eqref{eq:concise_core_block}.  The $6\times28$ core in the upper-left corner is marked in green; the $18\times18$ identity matrix occupies the lower-right corner.  The remaining nonzero entries are blue. Thus the 18 single-unit leaves are separated from the 28 core variables.

\subsection{From leaf recovery to core stability}

The goal is to make the APP parameters of the leaves vanish even though their outgoing messages remain channel-valued.  Define the neighbors of check $i$ by $\mathcal N_i=\{j:\omega_{ij}>0\}$. For $i\in\mathcal I_{\mathrm L}$, \eqref{eq:concise_maps} gives
\begin{equation}
\label{eq:revised_leaf_app}
\begin{aligned}
x_{i,j_i}^{(\ell)}&=\beta_{j_i},\\
z_{j_i}^{(\ell)}
&=\beta_{j_i}\left[
1-\prod_{j\in\mathcal N_i\setminus\{j_i\}}
(1-\omega_{ij}x_{ij}^{(\ell)})\right].
\end{aligned}
\end{equation}

Thus recovery of a leaf requires error-free extrinsic messages from \emph{all} its other check neighbors, not merely correct APP decisions at those variables.  The construction of $\boldsymbol{\Omega}_c$ retains every variable appearing in the product in \eqref{eq:revised_leaf_app}.

At least two units per core column sustain these zero messages: after a variable update excludes its target edge, at least one zero core factor remains.  This verifies the partial fixed point
\begin{equation*}
x_{ij}^*=0\ (j\in\mathcal{J}_c),\qquad
x_{i,j_i}^*=\beta_{j_i}\ (i\in\mathcal I_{\mathrm L}),
\end{equation*}
with $z_j^*=0$ for every variable.  The first equality includes messages from core variables to checks containing leaves.  Those checks therefore recover their leaves through \eqref{eq:revised_leaf_app}.

We now derive a condition for local convergence to this partial error-free fixed point.  The linearization is at zero in the \emph{core} coordinates, not at the origin of the full system. At the partial point, a check $p\in\mathcal I_{\mathrm L}$ sends a parameter $\beta_{j_p}$ toward each connected core variable.  These nonzero inputs contribute the factor
\begin{equation*}
\theta_j=\beta_j\prod_{p\in\mathcal I_{\mathrm L}}
(1-\omega_{pj}+\omega_{pj}\beta_{j_p}).
\end{equation*}
They cannot be discarded when testing core stability.

\begin{figure}
\centering
\includegraphics[page=5]{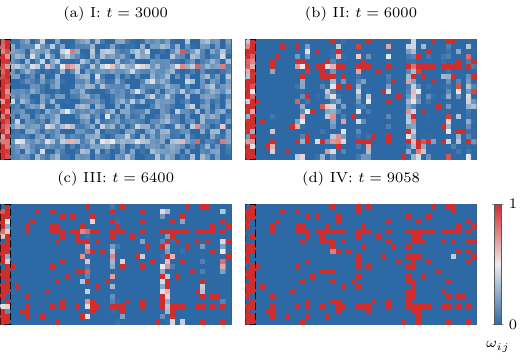}
\caption{Matrices from regimes I--IV on the same blue-to-red scale. Dashed boxes mark punctured columns.  The numbers of fractional entries are 890, 274, 96, and 2, respectively.}
\label{fig:learning_matrices}
\end{figure}

The core-message coordinates are indexed by
\begin{equation*}
\mathcal E_c=\{(i,j):i\in\mathcal I_c,\ j\in\mathcal{J}_c,\ \omega_{ij}>0\}.
\end{equation*}
Construct $\mathbf{C}_{c,0}$ and $\mathbf{V}_{c,0}$ by the same zero-point differentiation as above, using $\boldsymbol{\Omega}_c$ and the effective channel factors $\theta_j$. Applying \eqref{eq:concise_linear} to the core gives $\boldsymbol b_c=\boldsymbol 0$, since each core column contains at least two exact units.  Hence the first-order perturbations satisfy
\begin{equation*}
\delta\boldsymbol x_c^{(\ell+1)}
=\mathbf{J}_c\delta\boldsymbol x_c^{(\ell)},\qquad \mathbf{J}_c=\mathbf{V}_{c,0}\mathbf{C}_{c,0}.
\end{equation*}
Thus $\mathbf{J}_c$ is obtained from \eqref{eq:concise_j_entry} by restricting the edge indices to $\mathcal E_c$, the product to $s\in\mathcal I_c$, and replacing $\beta_j$ by $\theta_j$.  Three or more exact units in a core column make its output rows zero; exactly two can support linear feedback.

Leaf outputs are constant, and every other derivative outside the core block retains a zero core factor.  Hence the full active-edge Jacobian \emph{at the partial fixed point} is $\operatorname{diag}(\mathbf{J}_c,\mathbf 0)$ after reordering; it is not generally the origin Jacobian $\mathbf{J}$ in \eqref{eq:concise_linear}.  With channel dependence written as $\mathbf{J}_c(\gamma)$, the sufficient local stability condition is
\begin{equation*}
\rho(\mathbf{J}_c(\gamma))<1.
\end{equation*}
The corresponding threshold is
\begin{equation*}
\gamma_c^*=\inf\{\gamma\geq0:\rho(\mathbf{J}_c(\gamma))<1\}.
\end{equation*}
In \Fig{fig:learning_rho}, this threshold corresponds to the intersection of the curve with the horizontal black dashed line at $1$. The entries of $\mathbf{J}_c(\gamma)$ are nonincreasing in $\gamma$. If $\rho(\mathbf{J}_c(0))<1$, the stability condition holds for all $\gamma\geq0$, and $\gamma_c^*=0$. No core-based threshold is assigned unless the partial error-free fixed point exists. Local stability neither guarantees convergence from channel initialization nor determines the decoding threshold of the quantized SPA with $50$ iterations.

\subsection{Four regimes of the observed trajectory}

\begin{figure}
\centering
\includegraphics[page=6]{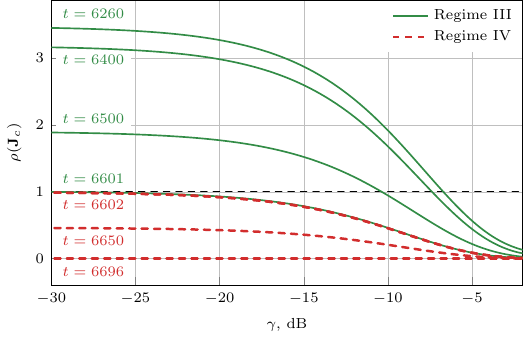}
\caption{Spectral radius in regimes III and IV. The transition between the two regimes occurs at epoch $6601$, when $\rho(\mathbf{J}_c) < 1$. From epoch $6696$ onward, the spectral radius is identically zero.}
\label{fig:learning_rho}
\end{figure}
We now examine the four regimes marked in \Fig{fig:learning_overview}. We relate the evolution of the BER under $50$-iteration SPA decoding to changes in the relaxed protograph and its local stability.

\paragraph*{I: Isolation-induced floor, epochs 0--5770} At least one column has no exact unit, as illustrated in \Fig{fig:learning_matrices}(a), so \eqref{eq:concise_floor} is positive. At epoch 1000 it accounts for $54.5\%$ of the saved BER (\Fig{fig:learning_overview}(I)). It is no longer dominant near the end: at epoch 5770 the isolation contribution is $1.26\times10^{-13}$ but BER is $6.05\times10^{-5}$. The last unsupported column acquires a unit at epoch 5771, removing this floor contribution exactly.

\paragraph*{II: Zero isolation floor without an error-free APP point, epochs 5771--6259} Every column now has a unit, but their arrangement, illustrated in~\Fig{fig:learning_matrices}(b), does not yet support an error-free APP fixed point.  This was checked in the full comparison recursion, not inferred solely from failure of the sufficient core test. The $18$-leaf block is fixed from epoch 5937; at epoch 6260 an additional unit finally gives every core column at least two units.  BER decreases from $6.02\times10^{-5}$ to $1.35\times10^{-9}$ during regime II (\Fig{fig:learning_overview}(II)), before the partial stability test becomes applicable.

\paragraph*{III: Partial error-free stability with a finite-SNR crossing, epochs 6260--6601} The $6\times28$ core and the $18$ leaves support the partial error-free fixed point (\Fig{fig:learning_matrices}(c)), while the BER continues to decrease (\Fig{fig:learning_overview}(III)). Representative spectral-radius curves are shown in green in \Fig{fig:learning_rho}. The intersection of each curve with the horizontal black dashed line at $1$ determines the corresponding local-stability threshold $\gamma_c^*$. Expressed in decibels as $10\log_{10}\gamma_c^*$, this threshold decreases from $-6.71$ to $-31.38$~dB over this regime. At the operating SNR of $10\log_{10}\gamma=-2.44$~dB, the spectral radius $\rho(\mathbf{J}_c(\gamma))$ is already below $1$ and decreases from $0.166$ to $0.014$. These crossings characterize local stability, not the decoding threshold; in particular, the corresponding local-stability thresholds lie below the rate-$1/2$ binary-input AWGN capacity limit of $-2.823$~dB.

\paragraph*{IV: Local stability for all SNR values, epochs 6602--9058} At epoch 6602, $\rho(\mathbf{J}_c(0))$ falls below one, so $\gamma_c^*=0$. Representative spectral-radius curves are shown in red in \Fig{fig:learning_rho}. From epoch 6696, the operator is nilpotent.  The final $\mathbf{J}_c$ is $81\times81$, with 79 nonzero entries and $\mathbf{J}_c^2\ne0$, $\mathbf{J}_c^3=0$; its zero spectral radius is structural. The final $\boldsymbol{\Omega}$ in \Fig{fig:learning_matrices}(d) contains 903 zeros, 199 units, and two fractional entries. The BER spikes in \Fig{fig:learning_overview}(IV) reflect SNR switching, not a loss of local stability.

The key distinction is therefore between eliminating the constant term in the APP update, which requires at least one unit entry per column, and eliminating the constant terms in all extrinsic-message updates, which requires at least two unit entries per column. A stable core relaxes the latter requirement for the full matrix: leaves whose columns contain only one unit entry can still be recovered through their APP messages. This explains the persistence of the core/leaf architecture. As optimization proceeds, almost all remaining weights concentrate at the boundaries. This observed boundary concentration is likely driven by the need to achieve a low output BER. The local stability analysis characterizes this evolution but neither proves the global optimality of the resulting structure nor determines the number of leaves.

\bibliographystyle{IEEEtran}
\bibliography{references}

\end{document}